\documentclass{article}

\usepackage{graphicx}
\usepackage{cmap}
\usepackage{tikz}
\usetikzlibrary{quantikz2}
\usepackage{algorithm}
\usepackage{algpseudocode}
\usepackage{amsthm}
\usepackage{amssymb}

\newtheorem{theorem}{Theorem}[section]
\newtheorem{corollary}{Corollary}[theorem]
\newtheorem{lemma}[theorem]{Lemma}

\title{Application of a Quantum Amplitude Redistribution Algorithm to the Data Filtering Problem}

\author{
K. R.~Zakharova, reshetova\_carina@yandex.ru
\\
A. A.~Chernikov, artem\_chernikov00@list.ru
\\
S. S.~Sysoev, s.s.sysoev@spbu.ru
\\
St.\ Petersburg State University, Russia
}
\begin{document}
\maketitle

\begin{abstract}
This paper presents an analysis of the applicability of a quantum amplitude redistribution algorithm to the data filtering problem and the results of modeling the algorithm's operation in comparison with a median filter.
\end{abstract}


\section{Introduction}

Quantum information studies the computational capabilities of physical systems described by the mathematical apparatus of quantum mechanics. One of the important and rapidly growing areas in this field is the development of new quantum algorithms in the circuit model \cite{deu}, \cite{dv}. The advantage of such algorithms is based on the use of quantum parallelism — the ability of a quantum system to follow several computational trajectories simultaneously. The number of such trajectories grows exponentially with the size of the computational system, but when obtaining classical information from a quantum state (upon measurement) it collapses to a single trajectory \cite{neu}. The main approach to overcoming this drawback is to redistribute the probability amplitudes of quantum states before their collapse (for example, using interference). This makes it possible to increase the probability of measuring certain states (more informative ones) relative to others (less informative ones). In some cases, the exponential advantage can be preserved, as for example in Shor's algorithm \cite{shor}.

The Grover algorithm \cite{grov} and its more general modifications \cite{brass} allow searching for numbers corresponding to an arbitrary algorithmic condition of the form $A(x) = 1$, where $A$ is some algorithm and $x$ is the numerical parameter of $A$. The choice of the algorithm $A$ determines the nature of the problem being solved. If, for example, $x$ is a numerical representation of a path in a complete graph and $A$ is a check of the condition $len(x) < b$, then the Grover algorithm can be used to solve the traveling salesman problem. In \cite{saf}, the algorithm \cite{brass} is used to amplify the amplitude of the correct result of binary classification in quantum machine learning.

If the condition $A(x) = 1$ corresponds to the situation of selecting the correct signal value from an array containing a digitized signal with outliers (anomalies), then the Grover algorithm can be used as a filtering method. The drawback of such a solution is that redistribution of probability amplitudes in the Grover algorithm requires a fixed number of steps (iterations) $T = O(\sqrt{N})$, where $N$ is the size of the input array.

In \cite{z1}, \cite{z2} a filtering algorithm is proposed that operates in time $O(1)$. Unlike \cite{grov}, \cite{brass}, in this algorithm the amplitudes of the vectors encoding the array elements are decreased proportionally to the distance of the elements from some reference value. As a result of measuring the state, one obtains the array index after probability-amplitude redistribution. Thus, at the output of the method, one obtains an element that was originally present in the data, which makes it similar to a median filter \cite{med}, \cite{med_21}, \cite{med_23}.

In this work, the properties of the algorithm \cite{z1} are investigated for the task of signal filtering; the dependence of filtering quality on certain characteristics of the signal and outliers is analyzed; and simulation results of the algorithm operation are presented in comparison with a median filter.
The paper is organized as follows: in the second section, a brief description of the algorithm from \cite{z1} is given. In the third section, the transformation scheme used for amplitude redistribution is presented. The fourth section is devoted to a mathematical analysis of filtering results under certain assumptions regarding anomalies encountered in the signal. In the fifth section, simulation results of the algorithm for filtering artifacts in real images are provided in comparison with a median filter. The conclusion formulates statements about the limits of applicability of the investigated algorithm.

\section{Quantum Amplitude Redistribution Algorithm}
The quantum amplitude redistribution algorithm (Quantum Amplitudes Redestribution Algorithm, abbreviated as QARA) described in \cite{z1}, \cite{z2} takes as input an array $A$ of $M$ $n$-bit non-negative integers and an $n$-bit reference value $r$.

To redistribute amplitudes, the algorithm uses a matrix parameterized by the angle $\varphi$:

\begin{equation}
R_n(\varphi) =
\left(\begin{array}{cccc}
\cos \frac{\varphi}{2} & \pm \frac{\sin \frac{\varphi}{2}}{\sqrt{2^m-1}} & \cdots & \pm \frac{\sin \frac{\varphi}{2}}{\sqrt{2^m-1}} \\
\pm \frac{\sin \frac{\varphi}{2}}{\sqrt{2^m-1}} & \cos \frac{\varphi}{2} & \cdots & \pm \frac{\sin \frac{\varphi}{2}}{\sqrt{2^m-1}} \\
\vdots & \vdots & \ddots & \vdots \\
\pm \frac{\sin \frac{\varphi}{2}}{\sqrt{2^m-1}} & \pm \frac{\sin \frac{\varphi}{2}}{\sqrt{2^m-1}} & \cdots & \cos \frac{\varphi}{2}
\end{array}\right)\label{eq:rot_matrix}
\end{equation}

The operation of the algorithm can be described by the following pseudocode:

\begin{algorithm}
\caption{Quantum Amplitude Redistribution Algorithm}\label{alg:qara}
    \begin{algorithmic}[1] 
        \Procedure{QARA}{$A,r$}
            \State \Comment{Initialize registers}
            \State Data register $|D\rangle \gets |0\rangle^{\otimes n}$
            \State Reference register $|R\rangle \gets r$
            \State $m \gets log_2 M$
            \State Counter register $|C\rangle \gets H^{\otimes m}|0\rangle^{\otimes m}$

            \State \Comment{Load data (Fig.~\ref{fig:arr_load})}
            \State $j=0$
            \While{$j<M$}
                \State // mark index $j$ in register $|C\rangle$:
                \State $\forall i: 2^i \land j = 0$ apply $X |C\rangle_i $ 
                \State // load $A_j$ into $|D\rangle$ controlled by $|C\rangle$:
                \State $\forall i: 2^i \land A_j = 1$ apply $|D\rangle \gets^{C} A_j$ 
                \State // unmark index $j$ in register $|C\rangle$:
                \State $\forall i: 2^i \land j = 0$ apply $X |C\rangle_i $ 
                \State $j \gets j+1$
            \EndWhile

            \State \Comment{Rotations}
            \State $i=0$
            \While{$i<n$}
                \State $\varphi_i = \frac{\pi}{2^{i+1}}$
                \State Apply $R_{\varphi_i}$ to $|C\rangle$, controlled by $|R\rangle_i$
                \State Apply $R_{-\varphi_i}$ to $|C\rangle$, controlled by $|D\rangle_i$
                \State $i \gets i+1$
            \EndWhile
            \State $j \gets Measure|C\rangle$
            \State \textbf{return} $A_j$
        \EndProcedure
    \end{algorithmic}
\end{algorithm}

\begin{figure}[ht]
\centering
\includegraphics[width=0.4\textwidth]{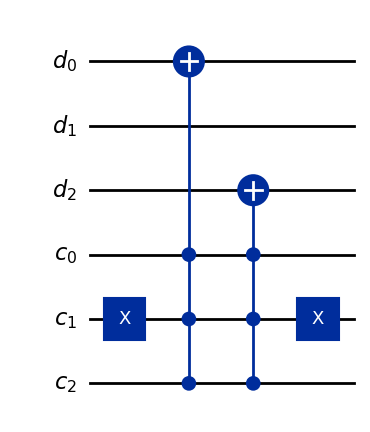}
\caption{Loading the value $A_5 = 5$. As a result, the index $|101\rangle$ in $|C\rangle$ is entangled with the value $|101\rangle$ in $|D\rangle$} \label{fig:arr_load}
\end{figure}

\begin{figure}[ht]
\centering
\includegraphics[width=1\textwidth]{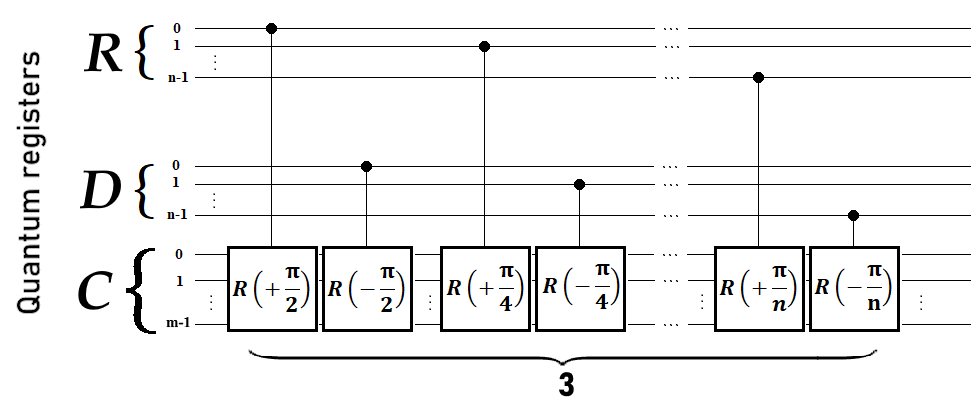}
\caption{Schematic diagram of the rotations of the quantum algorithm for probability redistribution}\label{fig:1_PSchm}
\end{figure}

\clearpage
\section{Implementation of the Rotation}
This section contains the definition of the transformation used in the algorithm (\ref{alg:qara}) for redistributing amplitudes between array indices, as well as a proof of its unitarity and a description of its decomposition into a sequence of elementary gates with a circuit depth that grows quadratically with the number of qubits.

\subsection{Definition}
The transformation is defined by the matrix (\ref{eq:rot_matrix}), $2^n \times 2^n$, where $n$ is the number of qubits.
An arbitrary placement of signs outside the diagonal is allowed, subject only to the requirement that the transformation is unitary. Among admissible sign arrangements, it was decided to use the one that proved simplest to implement.

We define the sign arrangement for one qubit as follows:
$$
R_1(\varphi) = 
\left(\begin{array}{cc}
\cos{\frac{\varphi}{2}} & -\sin{\frac{\varphi}{2}}\\
\sin{\frac{\varphi}{2}} & \cos{\frac{\varphi}{2}}\\
\end{array}\right)
$$

For $(n+1)$ qubits we define the sign arrangement recursively:
$$
R_{n+1}(\varphi) = 
\left(\begin{array}{cc}
A_n & -B_n\\
B_n & A_n\\
\end{array}\right)
$$

In the submatrix $A_n$ the signs are arranged in the same way as in $R_n$, while in the submatrix $B_n$ they are arranged as in the $n$-qubit Walsh--Hadamard transformation. Note that transposition of the matrix is equivalent to changing the signs of the elements outside the diagonal.

\subsection{Proof of Unitarity}

\begin{lemma}
Let $H_n$ be the Walsh--Hadamard operator for $n$ qubits. Then the following equality holds:
\begin{equation}
    H_nR_n(\varphi)H_n = R^\dagger_n(\varphi)
    \label{eq:1}
\end{equation}
\label{lemma:hrhr}
\end{lemma}
\begin{proof}
We use mathematical induction over the number of qubits. For one qubit we have:
$$
H_1R_1(\varphi)H_1 =
$$
$$
=\frac{1}{\sqrt{2}}
\left(\begin{array}{cc}
1 & 1\\
1 & -1\\
\end{array}\right)
\left(\begin{array}{cc}
\cos{\frac{\varphi}{2}} & -\sin{\frac{\varphi}{2}}\\
\sin{\frac{\varphi}{2}} & \cos{\frac{\varphi}{2}}\\
\end{array}\right)
\frac{1}{\sqrt{2}}
\left(\begin{array}{cc}
1 & 1\\
1 & -1\\
\end{array}\right) =
$$
$$
= 
\left(\begin{array}{cc}
\cos{\frac{\varphi}{2}} & \sin{\frac{\varphi}{2}}\\
-\sin{\frac{\varphi}{2}} & \cos{\frac{\varphi}{2}}\\
\end{array}\right) = R^\dagger_1(\varphi)
$$

Introduce the notation:
$$
a = \sqrt{\frac{2^n}{2^{n+1}-1}}\sin{\frac{\varphi}{2}}
$$

Then $R_{n+1}(\varphi)$ can be represented in the following form:
$$
R_{n+1}(\varphi) = 
\left(\begin{array}{cc}
\sqrt{1 - a^2}R_n & -aH_n\\
aH_n & \sqrt{1 - a^2}R_n\\
\end{array}\right)
$$

The operator $R_n$ is parameterized by some other angle different from $\varphi$; its value is omitted for brevity.

Induction step:
$$
H_{n+1}R_{n+1}(\varphi)H_{n+1} =
$$
$$
=
\frac{1}{\sqrt{2}}
\left(\begin{array}{cc}
H_n & H_n\\
H_n & -H_n\\
\end{array}\right)
\left(\begin{array}{cc}
\sqrt{1 - a^2}R_n & -aH_n\\
aH_n & \sqrt{1 - a^2}R_n\\
\end{array}\right) \times
$$
$$
\times\frac{1}{\sqrt{2}}
\left(\begin{array}{cc}
H_n & H_n\\
H_n & -H_n\\
\end{array}\right) =
$$
$$
= 
\left(\begin{array}{cc}
\sqrt{1 - a^2}R^\dagger_n & aH_n\\
-aH_n & \sqrt{1 - a^2}R^\dagger_n\\
\end{array}\right) = R^\dagger_{n+1}(\varphi)
$$
\end{proof}

\begin{theorem}
The transformation $R_n(\varphi)$ is unitary.
\end{theorem}
\begin{proof}
Again, we apply mathematical induction over the number of qubits. $R_1(\varphi)$ is nothing other than a rotation on the Bloch sphere around the $Y$ axis, whose unitarity is evident.
Let $b=1-a^2$.

Induction step:
$$
R_{n+1}R_{n+1}^\dagger = 
$$
$$
=\left(\begin{array}{cc}
\sqrt{b}R_n & -aH_n\\
aH_n & \sqrt{b}R_n\\
\end{array}\right)
\left(\begin{array}{cc}
\sqrt{b}R_n^\dagger & aH_n^\dagger\\
-aH_n^\dagger & \sqrt{b}R_n^\dagger\\
\end{array}\right) =
$$
$$
=
\left(\begin{array}{cc}
bR_nR_n^\dagger + a^2 H_n H_n^\dagger & a\sqrt{b}R_nH_n^\dagger - a\sqrt{b}H_nR_n^\dagger\\
a\sqrt{b}H_nR_n^\dagger - a\sqrt{b}R_nH_n^\dagger & bR_nR_n^\dagger + a^2 H_n H_n^\dagger\\
\end{array}\right)
$$

By the induction hypothesis, $R_nR_n^\dagger = I$, hence
$$bR_nR_n^\dagger + a^2 H_n H_n^\dagger = b I + a^2 I = I$$

From Lemma \ref{lemma:hrhr} it follows that $H_nR^\dagger_n = R_nH^\dagger_n$, therefore
$$a\sqrt{b}H_nR_n^\dagger - a\sqrt{b}R_nH_n^\dagger = 0$$

Hence, $R_{n+1}R_{n+1}^\dagger = I$. Similarly one can show that $R^\dagger_{n+1}R_{n+1} = I$.
\end{proof}

\subsection{Decomposition into a Sequence of Gates}
To use the transformation $R_n(\varphi)$ in practice, for example on a quantum computer simulator, one needs to decompose it into a sequence of elementary operators (gates). For convenience, we will refer to such a decomposition as a quantum algorithm that implements the operator $R_n(\varphi)$.

Find the basis in which the operator $R_{n+1}(\varphi)$ takes a simple form. First, we move to the basis with the transition matrix $CH_{n}$ --- a controlled Walsh--Hadamard operator controlled by the most significant qubit:
$$
CH_{n+1}R_{n+1}(\varphi)CH_{n+1} =
$$
$$
=
\left(\begin{array}{cc}
I & 0\\
0 & H_n\\
\end{array}\right)
\left(\begin{array}{cc}
\sqrt{1 - a^2}R_n & -aH_n\\
aH_n & \sqrt{1 - a^2}R_n\\
\end{array}\right)
\left(\begin{array}{cc}
I & 0\\
0 & H_n\\
\end{array}\right) =
$$
$$
=
\left(\begin{array}{cc}
\sqrt{1 - a^2}R_n & -aI\\
aI & \sqrt{1 - a^2}R^\dagger_n\\
\end{array}\right)
$$

Next, we move to another basis with a transition matrix equivalent to a rotation around the $Y$ axis on the Bloch sphere by an angle $\frac{\pi}{2}$:
$$
\frac{1}{2}
\left(\begin{array}{cc}
I & -I\\
I & I\\
\end{array}\right)
\left(\begin{array}{cc}
\sqrt{b}R_n & -aI\\
aI & \sqrt{b}R^\dagger_n\\
\end{array}\right)
\left(\begin{array}{cc}
I & I\\
-I & I\\
\end{array}\right) =
$$
$$ =
\left(\begin{array}{cc}
\frac{\sqrt{b}}{2}(R_n + R^\dagger_n) & \frac{\sqrt{b}}{2}(R_n - R^\dagger_n) - aI\\
\frac{\sqrt{b}}{2}(R_n - R^\dagger_n) + aI & \frac{\sqrt{b}}{2}(R_n + R^\dagger_n)\\
\end{array}\right)
$$

We have previously established that transposition of $R_n$ is equivalent to changing the signs of the off-diagonal elements; therefore $(R_n + R^\dagger_n)$ is a diagonal matrix, and $(R_n - R^\dagger_n)$ is a matrix with zeros on the diagonal. After several simplifications, we obtain:
$$
\left(\begin{array}{cc}
I\cos{\frac{\varphi}{2}} & -R^\dagger_n(\psi)\sin{\frac{\varphi}{2}}\\
R_n(\psi)\sin{\frac{\varphi}{2}} & I\cos{\frac{\varphi}{2}}\\
\end{array}\right)
$$

Here
$$
\sin{\frac{\psi}{2}} = \frac{1}{\sqrt{2^{n+1}-1}},
$$
$$
R_n(\psi) = 
\frac{1}{\sqrt{2^{n+1}-1}}
\left(\begin{array}{cccc}
\sqrt{2^n} & \pm 1 & \cdots & \pm 1\\
\pm 1 & \sqrt{2^n} & \cdots & \pm 1\\
\vdots & \vdots & \ddots & \vdots\\
\pm 1 & \pm 1 & \cdots & \sqrt{2^n}\\
\end{array}\right)
$$

Finally, we move to the basis with the transition matrix of the operator $R^\dagger_n(\psi)$ controlled by the most significant qubit:
$$
\left(\begin{array}{cc}
I & 0\\
0 & R^\dagger_n(\psi)\\
\end{array}\right)
\left(\begin{array}{cc}
I\cos{\frac{\varphi}{2}} & -R^\dagger_n(\psi)\sin{\frac{\varphi}{2}}\\
R_n(\psi)\sin{\frac{\varphi}{2}} & I\cos{\frac{\varphi}{2}}\\
\end{array}\right)\times
$$
$$
\times\left(\begin{array}{cc}
I & 0\\
0 & R_n(\psi)\\
\end{array}\right) =
\left(\begin{array}{cc}
I\cos{\frac{\varphi}{2}} & -I\sin{\frac{\varphi}{2}}\\
I\sin{\frac{\varphi}{2}} & I\cos{\frac{\varphi}{2}}\\
\end{array}\right)
$$

Thus, we find the basis in which the original operator $R_{n+1}(\varphi)$ acts as a rotation around the $Y$ axis on the Bloch sphere for the most significant qubit by the angle $\varphi$, implemented by just a single gate. To carry out the transformation, it is sufficient to switch to this basis, perform the rotation, and return to the initial basis, as illustrated in Fig.~\ref{fig:rnphi_prototype}.
\begin{figure*}[t]
    \centering
\begin{quantikz}[wire types={b,q},classical gap=0.1cm]
&\gate[]{H_n}\gategroup[2,steps=3,style={dashed,rounded corners}]{}&&\gate[]{R_n^\dagger(\psi)}&&\gate[]{R_n(\psi)}\gategroup[2,steps=3,style={dashed,rounded corners}]{}&&\gate[]{H_n}& \\
&\ctrl{-1}&\gate[]{R_y(\frac\pi2)}&\ctrl{-1}&\gate[]{R_y(\varphi)}&\ctrl{-1}&\gate[]{R^\dagger_y(\frac\pi2)}&\ctrl{-1}&
\end{quantikz}
    \caption{The first step in finding the decomposition of $R_{n+1}(\varphi)$. The dashed lines highlight transition blocks to the appropriate basis. Here $\sin{\frac{\psi}{2}} = \frac{1}{\sqrt{2^{n+1}-1}}$.}
    \label{fig:rnphi_prototype}
\end{figure*}
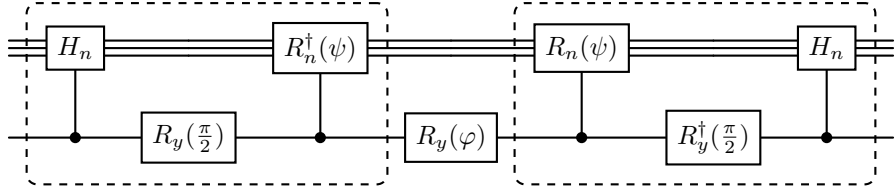

By recursively expanding $R_n(\psi)$ using the same scheme, we obtain an exponential growth of circuit depth with increasing number of qubits. However, this algorithm can be substantially simplified after decomposing the two internal operators $R_n(\psi)$. Between the transition back to the initial basis in the first transformation and the transition to the required basis in the second, no other gates occur. Therefore, these two transitions can be omitted. Carrying out such simplifications recursively and expanding all subsequent $R_k$, we obtain the optimized circuit shown in Fig.~\ref{fig:rnphi_final}.

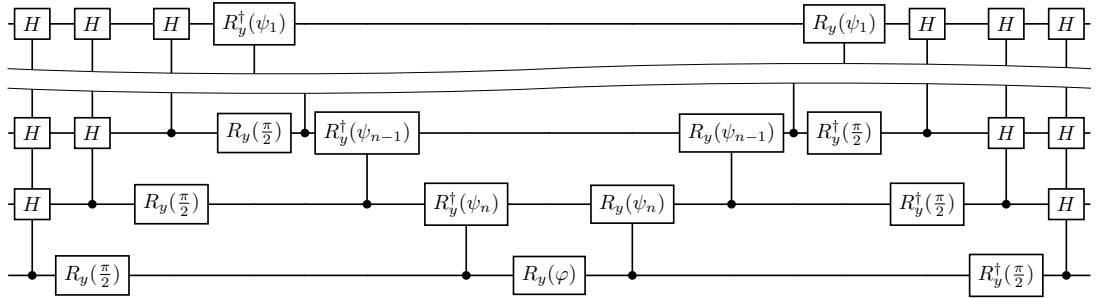
\begin{figure*}[t]
    \centering
\scalebox{0.78}{
\begin{quantikz}[wire types={q,q,q,q,q},classical gap=0.1cm,column sep=0.1cm]
&\gate{H}&\gate{H}&\gate{H}&\gate{R^\dagger_y(\psi_1)}&&&&&&&&\gate{R_y(\psi_1)}&\gate{H}&\gate{H}&\gate{H}& \\
\wave&&&&\ctrl{-1}&&&&&&&&\ctrl{-1}&&&&\\
&\gate{H}&\gate{H}&\ctrl{-2}&\gate{R_y(\frac\pi2)}&\ctrl{-1} &\gate{R^\dagger_y(\psi_{n-1})}&&&&\gate{R_y(\psi_{n-1})}&\ctrl{-1}&\gate{R^\dagger_y(\frac\pi2)}&\ctrl{-2}&\gate{H}&\gate{H}&\\
&\gate{H}&\ctrl{-3}&\gate{R_y(\frac\pi2)}&&&\ctrl{-1} &\gate{R^\dagger_y(\psi_n)}&&\gate{R_y(\psi_n)}&\ctrl{-1}&&&\gate{R^\dagger_y(\frac\pi2)}&\ctrl{-3}&\gate{H}& \\
&\ctrl{-4}&\gate{R_y(\frac\pi2)}&&&&&\ctrl{-1}&\gate{R_y(\varphi)}&\ctrl{-1}&&&&&\gate{R^\dagger_y(\frac\pi2)}&\ctrl{-4}&
\end{quantikz}
}
    \caption{Decomposition of $R_{n+1}(\varphi)$. Here $\sin{\frac{\psi_k}{2}} = \frac{1}{\sqrt{2^{k+1}-1}}$.}
    \label{fig:rnphi_final}
\end{figure*}

As a result, the number of elementary gates implementing the required algorithm becomes asymptotically $O(n^2)$, where $n$ is the number of qubits on which the transformation is defined. Thus, it depends quadratically on the number of qubits, which is a quite acceptable result. If in one time step we allow the simultaneous execution of several single-qubit Walsh--Hadamard operators controlled by the same qubit, then the circuit depth grows linearly with the number of qubits.

\section{Analysis of Applicability of the Algorithm to the Data Filtering Problem}

\subsection{Probability Distribution After Applying the Algorithm}

Before applying the algorithm (\ref{alg:qara}), the elements of the input array in register $D$ are pairwise entangled with their indices in register $C$, and the probability distribution among them is uniform. After applying the algorithm, the amplitudes in register $C$ are redistributed and the pairwise entanglement with the array elements disappears. To analyze the possibility of using the algorithm as a filter, we consider how the probability of an arbitrary index $C_k$ changes:

\begin{theorem}\label{t_ck}
If all elements in the array $D$ are distinct, then the probability of the element with index $k$ when measuring register $C$ after a single application of the algorithm (\ref{alg:qara}) is given by:
\begin{equation}             \label{eq:2}
 \begin{aligned}
 P(C_k) = \frac{1}{M}[\cos^2(\sum_{i=0}^{n-1}\frac{\pi(r_i-d_i^k)}{2^{n-i+1}}) +\\+ \frac{1}{M-1}\sum_{j \in [0,M-1], j\neq k}^{}\sin^2(\sum_{i=0}^{n-1}\frac{\pi(r_i- d_i^j)}{2^{n-i+1}})],
  \end{aligned}
\end{equation}
\end{theorem}
where $M$ is the number of elements of the input array,\\
$n$ is the number of qubits required to store one element,\\
$r \in \{0,1\}^n$ is the reference value,\\
the upper indices $k$ and $j$ correspond to the element numbers $d$ of the input array,\\
the lower indices $i$ correspond to the bit numbers starting from the least significant bit.\\
Neglecting the schematic details of the bitwise execution of the algorithm, the formula \ref{eq:2} can be rewritten as:
\begin{equation}             \label{eq:2_1}
 \begin{aligned}
 P(C_k) = \frac{1}{M}[\cos^2(\frac{\pi}{2^{n+1}}(r-d^k)) +\\+ \frac{1}{M-1}\sum_{j \in [0,M-1], j\neq k}^{}\sin^2(\frac{\pi}{2^{n+1}}(r-d^j))],
  \end{aligned}
\end{equation}

\begin{proof}
The algorithm (\ref{alg:qara}) consists of sequentially applying the matrix $R_n(\varphi)$ (rotations) to register $C$, controlled by the values of qubits in register $D$. The angle $\varphi$ for each rotation is determined by the index of the control qubit $0 \leq i < n$:
$$
\varphi_i = \pi/2^{i+1}
$$
In this case:
\begin{enumerate}
    \item if $r_i > d_i^k$, then a rotation by $\varphi_i$ is applied to the vector $|C_k\rangle$,
    \item if $r_i < d_i^k$, then a rotation by $-\varphi_i$ is applied to the vector $|C_k\rangle$,
    \item if $r_i = d_i^k$, then no rotation is applied to the vector $|C_k\rangle$.
\end{enumerate}
This rule can be generalized by the following expression for the rotation angle of the vector $|C_k\rangle$:
$$
\varphi_i^k = \varphi_i (r_i - d_i^k)
$$

From the decomposition of the rotation matrix (\ref{eq:rot_matrix}) into a sequence of elementary gates in Fig.~\ref{fig:rnphi_final}, it follows that applying the rotations sequentially is equivalent to performing a single rotation by their total angle:
$$
\prod_i R_n( \varphi_i) = R_n(\sum_i \varphi_i)
$$
The most significant qubit $i=n-1$ controls the rotation by the largest angle $\pi/4$, while the least significant qubit ($i=0$) controls the rotation by the smallest angle $\pi/2^{n+1}$. For the vector $|C_k\rangle$, the total angle $\varphi_k$ is:
\begin{equation}
    \varphi_k = \sum_{i=0}^{n-1} \frac{\pi(r_i - d_i^k)}{2^{n-i+1}}
\end{equation}
Thus, after applying all rotations, the amplitude of the vector $|C_k\rangle$ is distributed among all vectors in the space of register $C$:  
\begin{equation}
\begin{aligned}
|C_k\rangle \rightarrow cos(\varphi_k) |C_k\rangle  \\
|C_j\rangle_{j \neq k} \rightarrow \pm\frac{sin(\varphi_j)}{\sqrt{M-1}} |C_k\rangle     
\end{aligned}
\end{equation}
The rotations applied to vectors $C_{j, j\neq k}$ also redistribute the amplitudes among all vectors in the space of register $C$. At the same time, no interference occurs between vectors with different indices $j$, since they are entangled with different values in register $D$. Therefore, from each such vector, the value $a_j$ is added to the amplitude of $|C_k\rangle$:
$$
a_j = \pm\frac{sin(\varphi_j)}{\sqrt{M-1}} 
$$
Taking into account the initial amplitudes of all vectors in the space, which are $1/\sqrt{M}$, for the probability $P(|C_k\rangle)$ we obtain:
\begin{equation}
    P(|C_k\rangle) = \frac{1}{M} (cos^2(\varphi_k) + \sum_{j\neq k} a_j^2)
\end{equation}
\end{proof}

\subsection{Dependence of Filtering Quality on Input Data}

Theorem \ref{t_ck} lets us formulate several useful corollaries regarding the applicability of the investigated algorithm.

Expression (\ref{eq:2}) determines the probability of obtaining the outlier index upon measurement only in the case when all values in the array $D$ are distinct. Repeated values in the input array lead to interference in registers $C$ and $D$, which can both decrease and increase this probability. For this reason, for practical filter implementations one can consider a modification of the algorithm in which the values of the array are loaded into register $D$ together with their indices (in the least significant bits). This modification requires increasing register $D$ by the size of register $C$. The added qubits do not participate in rotations and are needed only to prevent interference of indices in register $|C\rangle$ when the array values $A$ are identical.

From expression (\ref{eq:2}) it follows that the reduction of the outlier probability is influenced by two factors:
\begin{itemize}
    \item the number and significance of the bits in which the outlier differs from the reference value,
    \item the number of array elements that differ from the reference value (outliers).
\end{itemize}
The first factor determines how close the angle $\varphi_k$ under the cosine in the first part of expression (\ref{eq:2}) is to $\pi/2$, while the second factor determines how many non-zero terms will appear in the second part of the expression.

\begin{corollary}[Best Conditions for Filtering]
If the input array contains a single deviation (outlier) $d^k$:
$$
\forall i \quad  d_i^k > r_i
$$
and all other elements of the array match the reference value:
$$
\forall j \in [0, M-1], j\neq k \quad  d^j = r ,
$$
then the probability of obtaining the index $k$ when measuring register $C$ after applying the algorithm is the smallest among all possible input arrays and is equal to:
$$
P(|C_k\rangle) = \frac{1}{M}cos^2\Big(\frac{\pi}{2} - \frac{\pi}{2^{n+1}}\Big)
$$
\end{corollary}
\begin{proof}
Since the outlier $d^k$ is unique, its index in register $C$ does not interfere with other indices, and the cosine value in expression (\ref{eq:2}) remains valid for it. For the remaining indices, rotations are not performed due to the condition $d^j = r \quad \forall j\neq k$. Therefore, nothing is added to the amplitude of $|C_k\rangle$ from the other superposition vectors.
\end{proof}

\begin{theorem}\label{t_2s}
    If the input array contains a single outlier $d^k$, such that for some $l \in \mathbb{N}$:
$$
    d^k \geq 2^lr,
$$
and the deviation from $r$ for all other elements does not exceed some $\sigma$:
$$
\forall j\neq k \quad |d^j - r| \leq \sigma
$$
then the probability of obtaining the index $k$ when measuring register $C$ after applying the algorithm can be upper-bounded:
\begin{equation}\label{eq:3}
P(|C_k\rangle) \leq \frac{1}{M} \Big( cos^2 \Big( \frac{\pi}{4} - \frac{\pi}{2^{l+2}} \Big) + sin^2 \Big(\frac{\pi\sigma}{2^{n+1}} \Big) \Big)
\end{equation}
If $\sigma$ is also estimated in terms of $r$:
$$
\sigma \leq \frac{r}{2^p}
$$
then the upper bound for $P(|C_k\rangle)$ becomes:
\begin{equation}\label{eq:t3_est}
P(|C_k\rangle) \leq \frac{1}{M} \Big( cos^2 \Big( \frac{\pi}{4} - \frac{\pi}{2^{l+2}} \Big) + sin^2 \Big(\frac{\pi}{2^{l+p+1}} \Big) \Big)
\end{equation}

\end{theorem}

\begin{proof}
    The worst-case scenario that maximizes the probability $P(|C_k\rangle)$ under the conditions of the theorem is the following case:
    \begin{itemize}
        \item $r$ is a power of 2 ($r = 2^s$),
        \item $d^k = 2^l r$,
        \item $\forall j \neq k \quad d^j - r = \sigma = \frac{r}{2^p}$ or $r - d^j = \sigma = \frac{r}{2^p}$.
    \end{itemize}

    The first two conditions minimize the angle $\varphi_k$, which is the argument of the cosine in (\ref{eq:2}), since $\varphi_k = \frac{\pi}{2^{n+1}}(d^k - r)$. If $s$ is the position of the most significant bit in $r$, then
    $$
    d^k - r \geq 2^{l+1}r - r = r(2^{l+1} -1) \geq 2^s(2^{l+1}-1)
    $$
    The third condition maximizes the values of the angles $\varphi_j$, which are the arguments of the sines in (\ref{eq:2}). In this case,
    $$
    \frac{r}{2^p} \leq \frac{d_k}{2^{l+p}} \leq \frac{2^n}{2^{l+p}}
    $$
    Substituting these worst-case values into (\ref{eq:2}) and taking the most significant bit of $d^k$ as $n-1$ ($n-1=s+l$), we obtain:
    $$
    P(|C_k\rangle) \leq \frac{1}{M}\Big( cos^2\Big(\frac{\pi}{2^{n-(n-1)+1}} - \frac{\pi}{2^{n-(n-1-l)+1}}\Big) + 
    $$
    $$
    +\frac{1}{M-1}\sum_{j\neq k} sin^2 \Big(\frac{\pi\sigma}{2^{n+1}}\Big)\Big) \leq 
    $$
    $$
    \leq \frac{1}{M}\Big( cos^2\Big(\frac{\pi}{4} - \frac{\pi}{2^{l+2}}\Big) +  sin^2 \Big(\frac{\pi}{2^{l+p+1}}\Big)\Big)
    $$
\end{proof}

An important corollary of Theorems \ref{t_ck} and \ref{t_2s} is the following set of recommendations for using the algorithm:
\begin{itemize}
    \item The algorithm is sensitive to how fully populated the values of the most significant qubits are. The absence of differences between outliers and useful data in the most significant qubits leads to small values of the cosine argument $\varphi_k$ in (\ref{eq:2}), and consequently to an increase in the probability of obtaining the outlier index upon measurement.
    \item The algorithm is also sensitive to the filling of the least significant qubits. Shifting all array values by bits to the right increases the ratio $d^k/r$, which can significantly decrease the probability of measuring the outlier index in (\ref{eq:3}).
\end{itemize}

\section{Experimental Evaluation of the Algorithm}\label{exp}

The motivation for using the algorithm (\ref{alg:qara}) instead of a median filter is the gain in computational complexity. The median filter requires sorting the input array, which makes its complexity $O(n*M*log_2M)$, where $M$ is the array size and $n$ is the bit-width of the numbers compared during sorting. The complexity of the algorithm (\ref{alg:qara}) does not depend on the array size $M$ and is determined only by the bit-width: $O(n)$. Obtaining a reference value (for example, as the average value of the array elements) may require $O(n*M)$ bit-addition operations, leaving the gain of $log_2M$ times. If the reference value is taken as the value obtained after applying the algorithm to the previous window (the previous data segment), then computing the average is not required.

Adding a register with element indices that ensures the uniqueness of compared values does not increase computational complexity, since the added qubits do not participate in rotations.

The diagrams in Fig.~(\ref{fig:16_501510_plt}, \ref{fig:17_8329631424510_plt}) demonstrate the results of applying the algorithm to arrays of 4 and 8 elements. It can be seen that the algorithm orders the measurement probabilities of element indices in accordance with their values.

\begin{figure}[h!]
\centering
\includegraphics[width=1\textwidth]{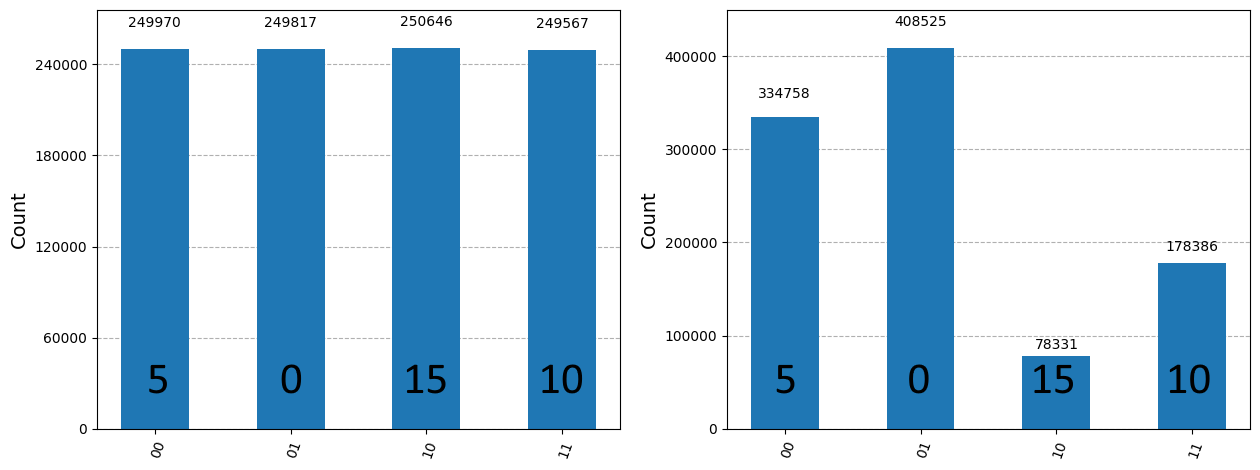}
\caption{Probabilities of array indices [5 0 15 10] for the reference value $R = 0$ before processing (left) and after applying the rotations (right), 1,000,000 measurements}\label{fig:16_501510_plt}
\end{figure}

\begin{figure*}[t]
\centering
\includegraphics[width=1\textwidth]{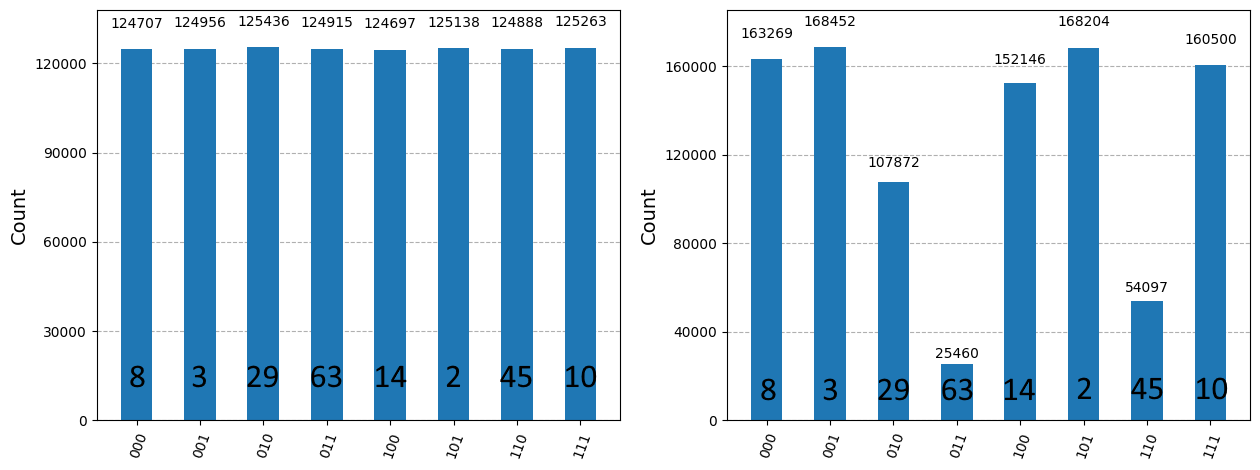}
\caption{Probabilities of array indices [8 3 29 63 14 2 45 10] for the reference value $R = 0$ before processing (left) and after applying the rotations (right), 1,000,000 measurements}\label{fig:17_8329631424510_plt}
\end{figure*}

Below are processing plots of a noisy triangular signal using a classical median filter and a quantum feedback filter, where the reference value (i.e., an analogue of the median) is the measurement result obtained in the previous cycle.

\begin{figure*}[t]
\centering
\includegraphics[width=1\textwidth]{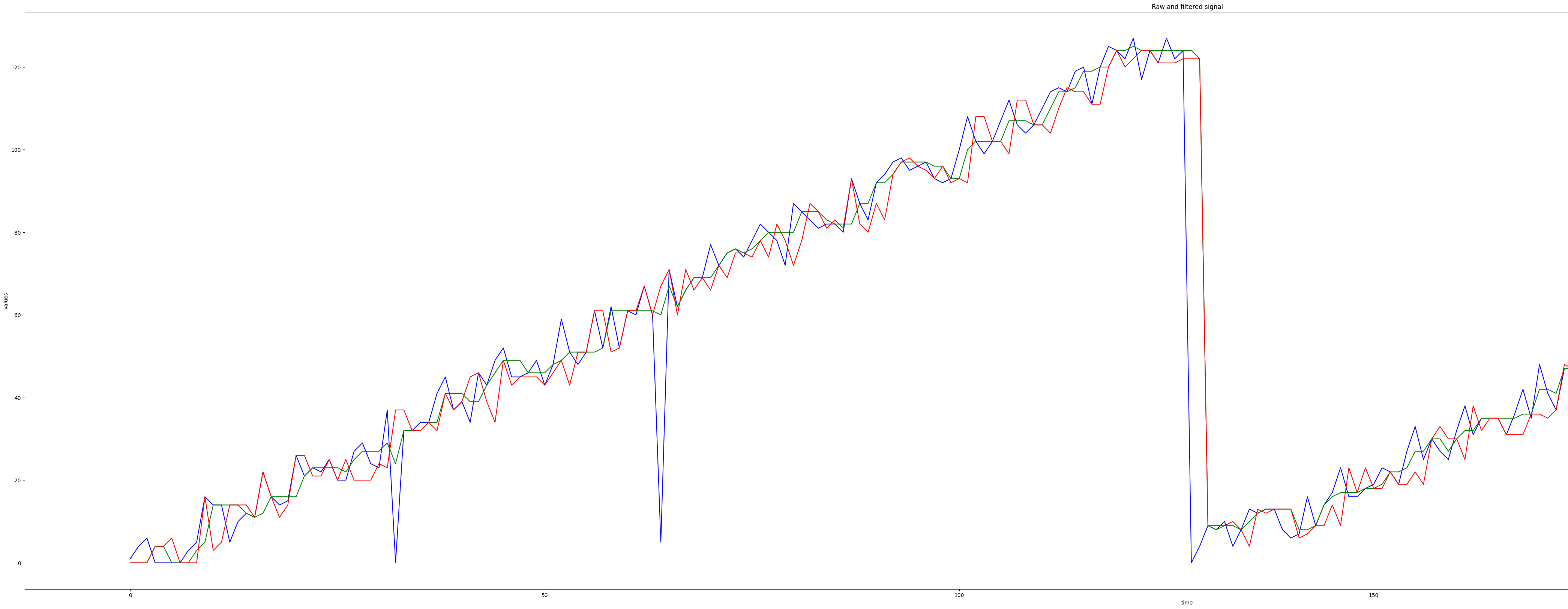}
\caption{Processing a noisy signal with two outliers (blue) using a median filter (green) and a quantum filter (red)}
\label{fig:18_two_outliers}
\end{figure*}

The following presents the results of processing various images with artificially introduced artifacts (highlighted in white) in comparison with a median filter. The settings of the quantum filter were optimized with respect to the maximum and minimum values within the window (array of pixels).

Elimination of repeated elements is achieved via additional $m = \log M$ qubits for each pixel, containing the element index in the array.

Image sources:\\
https://sipi.usc.edu/database/,
https://rentgenogram.com/.

\begin{figure*}[t]
\centering
\includegraphics[width=1\textwidth]{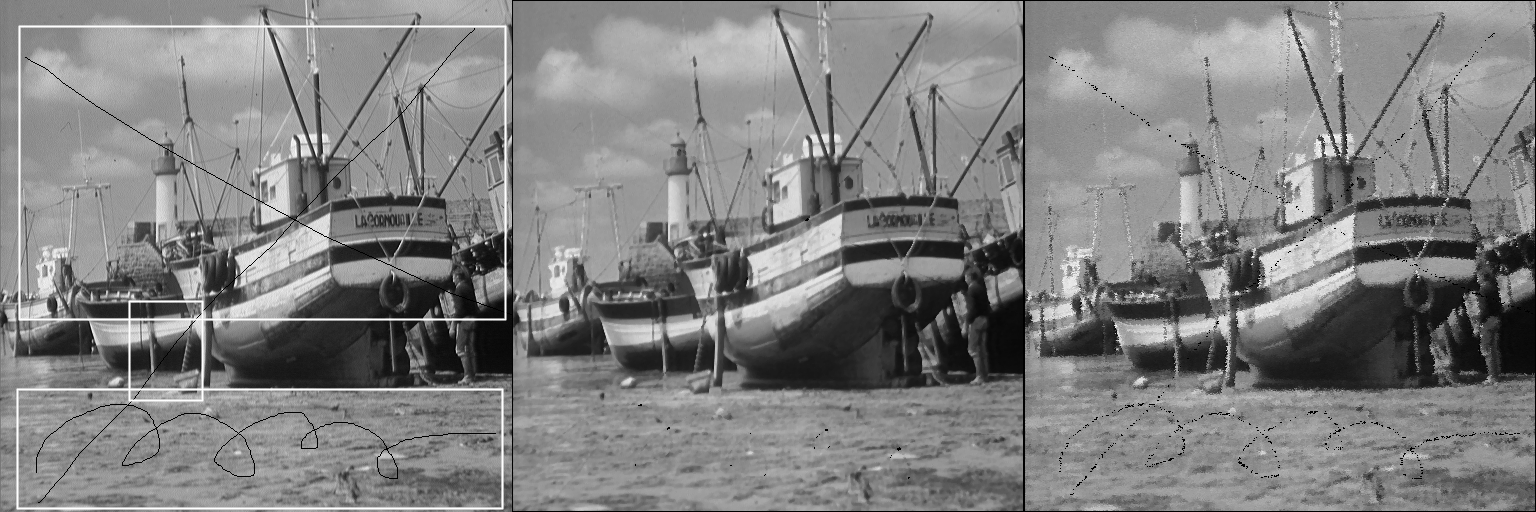}
\caption{Processing a black-and-white 8-bit image with an artifact in PNG format, 512*512 pixels. Original image (left), result of the median filter (middle), and result of the quantum feedback filter (right). Window width is 8 pixels.}\label{fig:boat_8}
\end{figure*}
\begin{figure*}[t]
\centering
\includegraphics[width=1\textwidth]{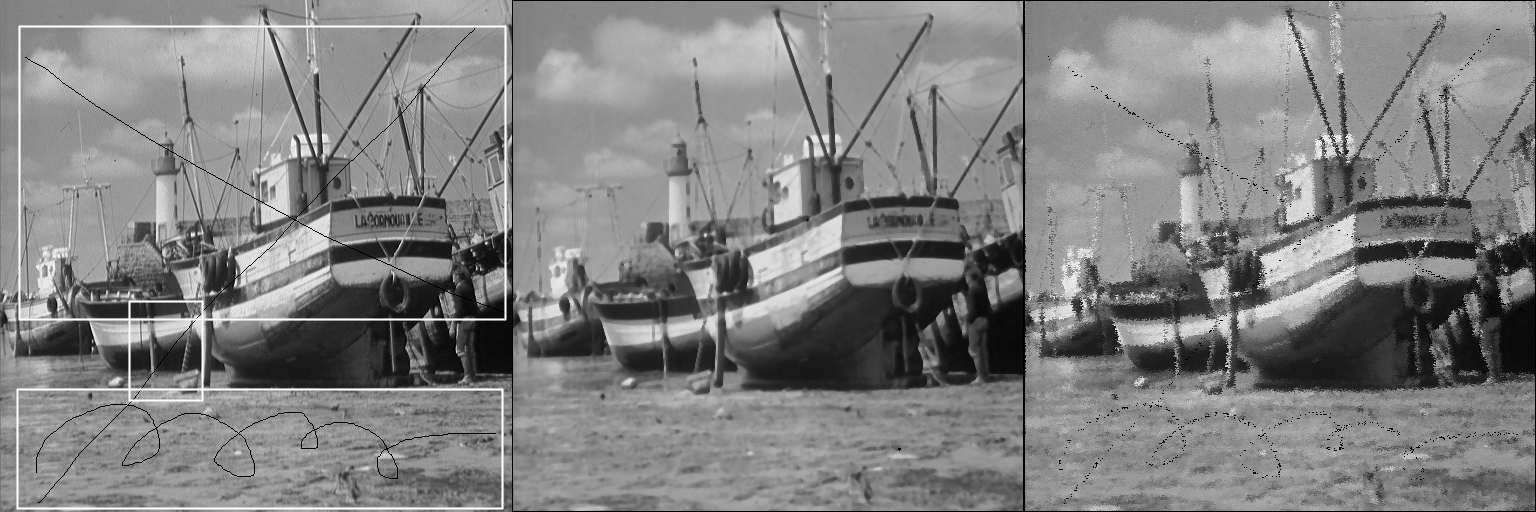}
\caption{Processing a black-and-white 8-bit image with an artifact in PNG format, 512*512 pixels. Original image (left), result of the median filter (middle), and result of the quantum feedback filter (right). Window width is 16 pixels.}\label{fig:boat_16}
\end{figure*}

\begin{figure*}[t]
\centering
\includegraphics[width=1\textwidth]{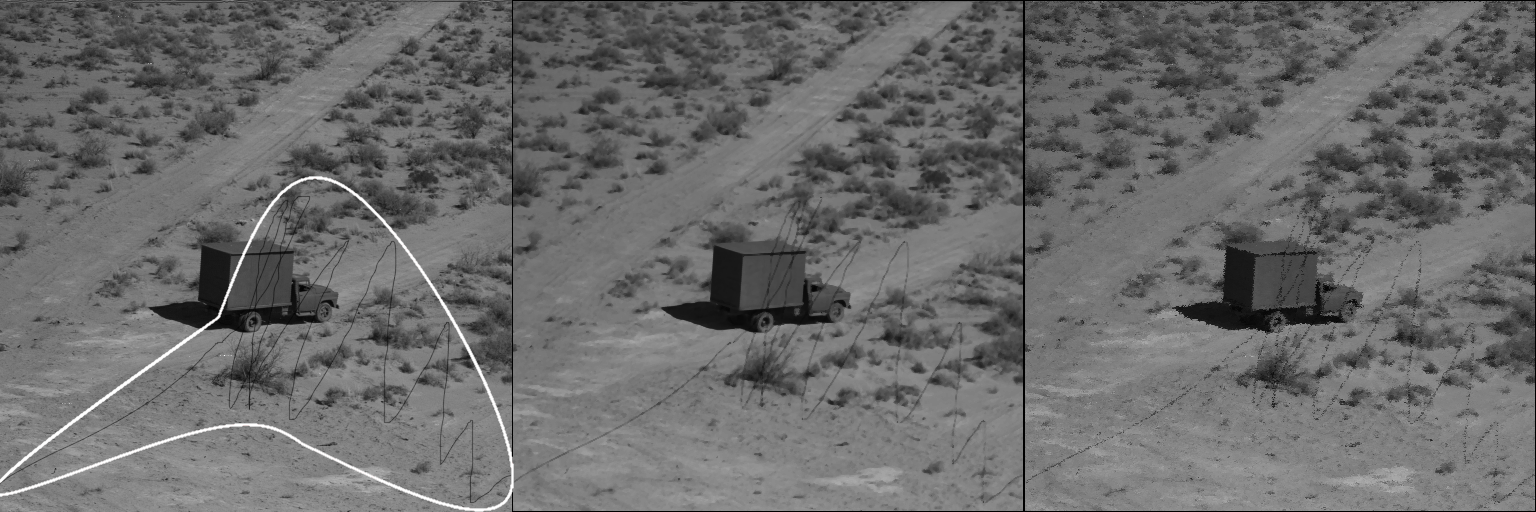}
\caption{Processing a black-and-white 8-bit image with an artifact in TIFF format, 512*512 pixels. Original image (left), result of the median filter (middle), and result of the quantum feedback filter (right). Window width is 8 pixels.}\label{fig:truck_8}
\end{figure*}
\begin{figure*}[t]
\centering
\includegraphics[width=1\textwidth]{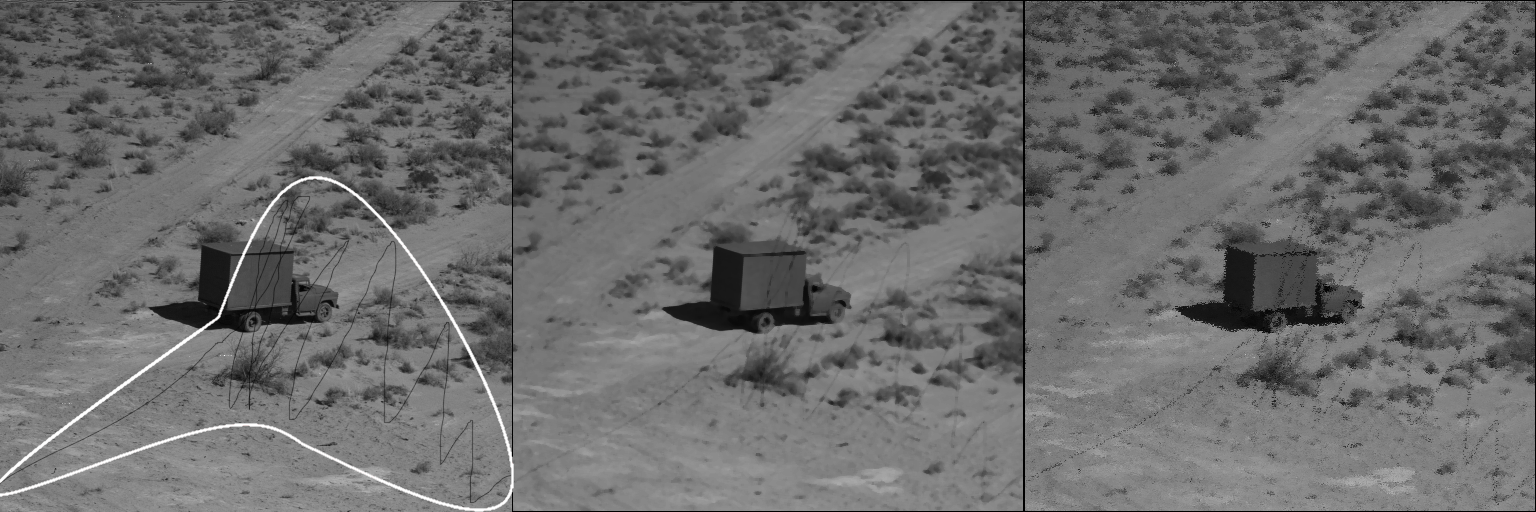}
\caption{Processing a black-and-white 8-bit image with an artifact in TIFF format, 512*512 pixels. Original image (left), result of the median filter (middle), and result of the quantum feedback filter (right). Window width is 16 pixels.}\label{fig:truck_16}
\end{figure*}

\begin{figure*}[t]
\centering
\includegraphics[width=1\textwidth]{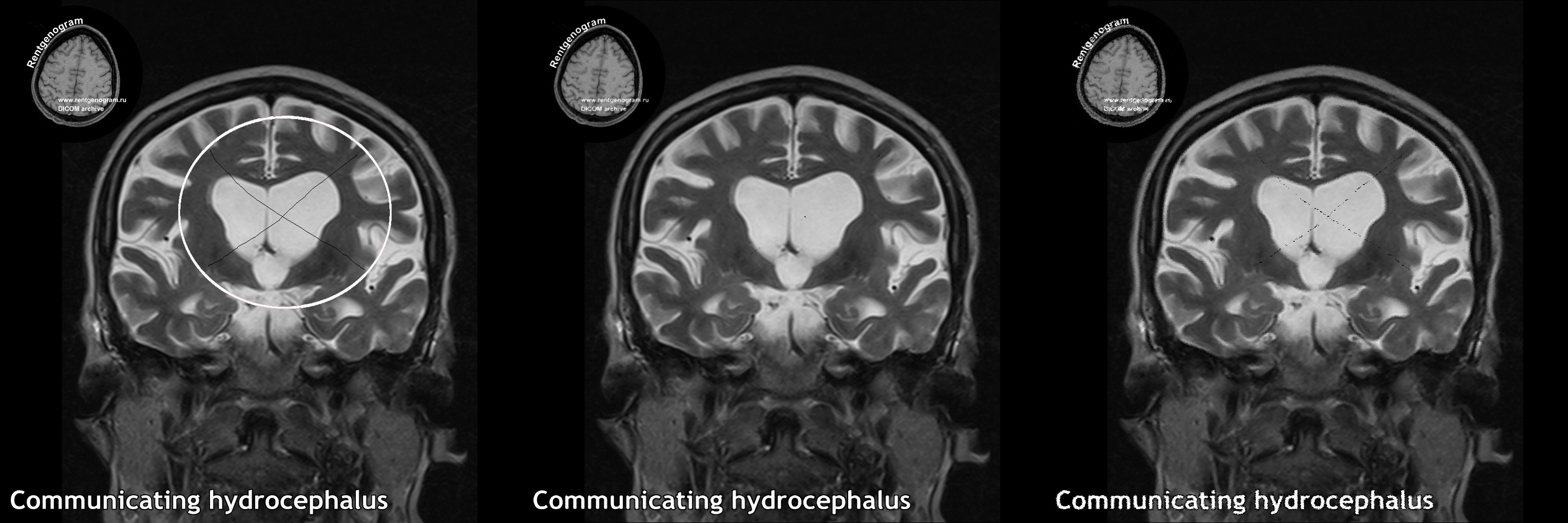}
\caption{Processing a black-and-white 8-bit MRI image with an artifact in JPG format, 1018*1018 pixels. Original image (left), result of the median filter (middle), and result of the quantum feedback filter (right). Window width is 8 pixels.}\label{fig:MRI_1018_8}
\end{figure*}
\begin{figure*}[t]
\centering
\includegraphics[width=1\textwidth]{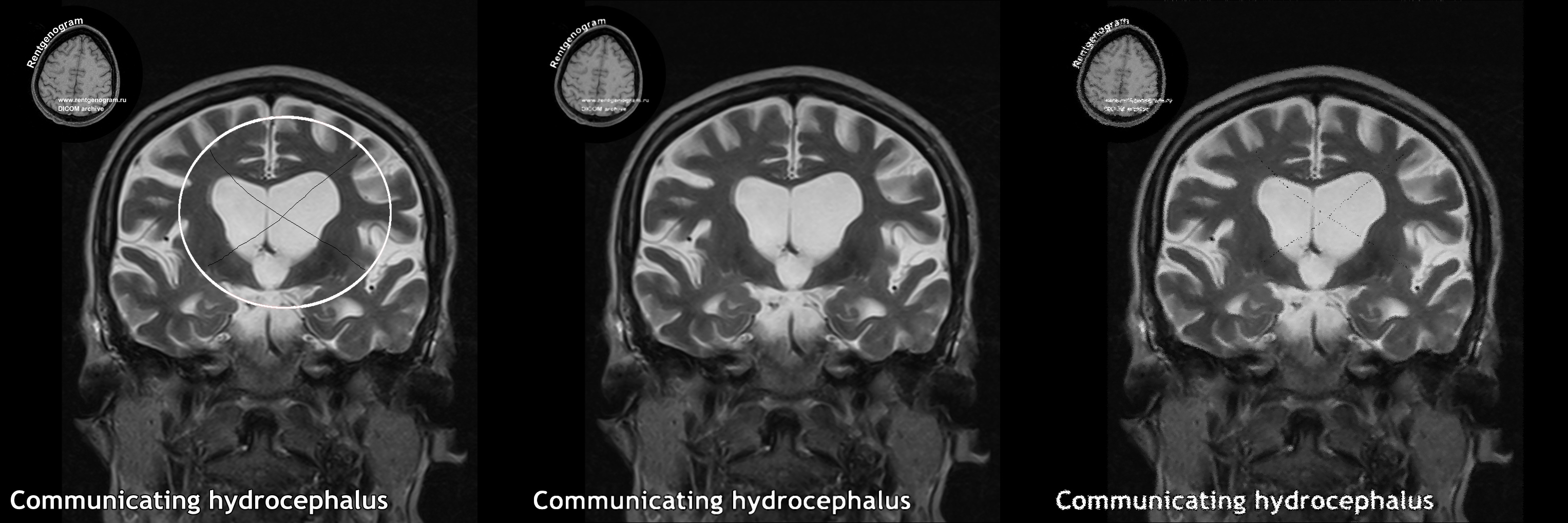}
\caption{Processing a black-and-white 8-bit MRI image with an artifact in JPG format, 1018*1018 pixels. Original image (left), result of the median filter (middle), and result of the quantum feedback filter (right). Window width is 16 pixels.}\label{fig:MRI_1018_16}
\end{figure*}

\begin{figure*}[t]
\centering
\includegraphics[width=1\textwidth]{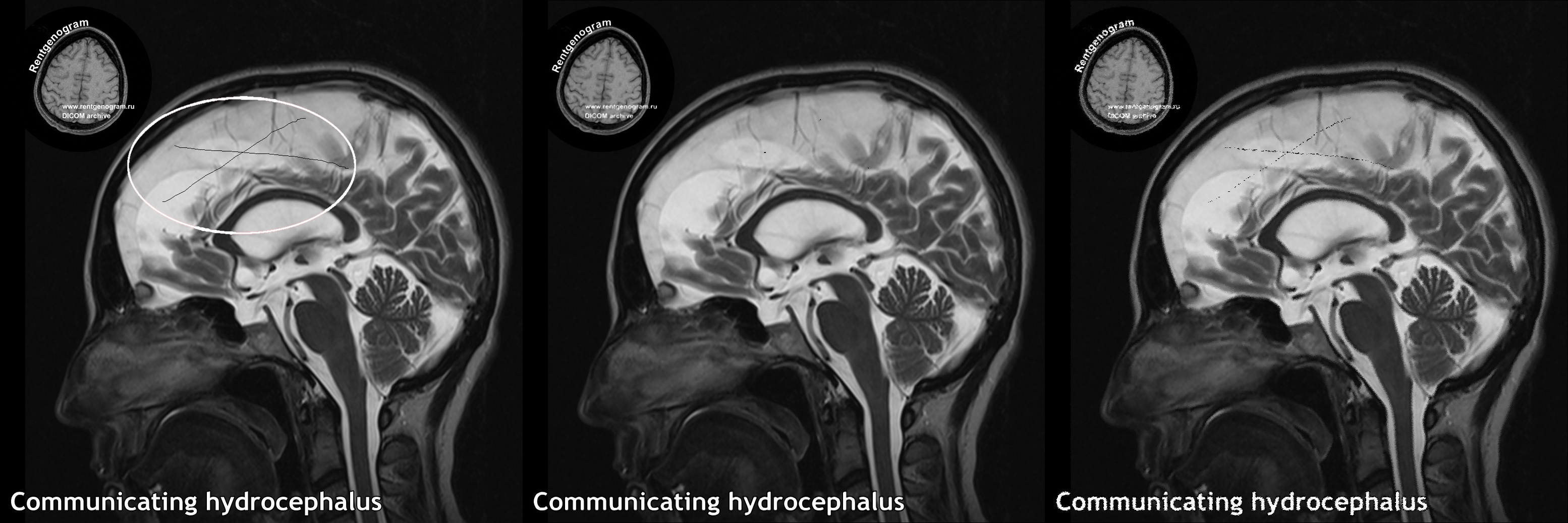}
\caption{Processing a black-and-white 8-bit MRI image with an artifact in JPG format, 1036*1036 pixels. Original image (left), result of the median filter (middle), and result of the quantum feedback filter (right). Window width is 8 pixels.}\label{fig:MRI_1036_8}
\end{figure*}
\begin{figure*}[t]
\centering
\includegraphics[width=1\textwidth]{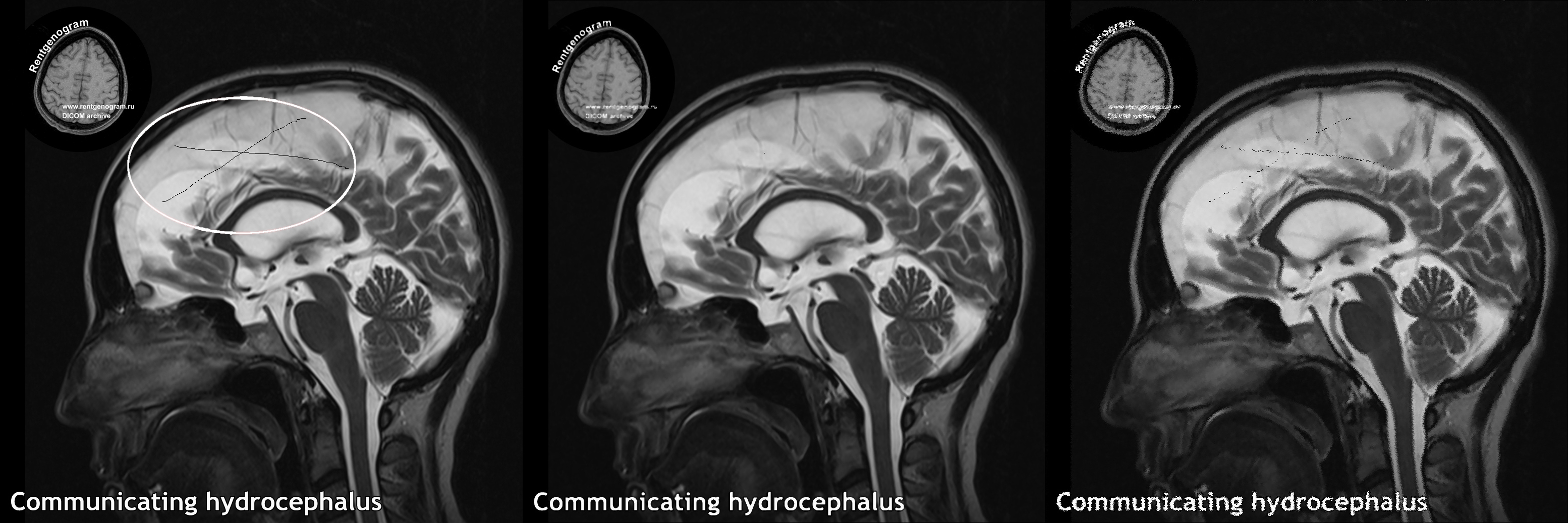}
\caption{Processing a black-and-white 8-bit MRI image with an artifact in JPG format, 1036*1036 pixels. Original image (left), result of the median filter (middle), and result of the quantum feedback filter (right). Window width is 16 pixels.}\label{fig:MRI_1036_16}
\end{figure*}

\clearpage
\section{Conclusion}

As a result of the theoretical and experimental study of the algorithm presented in \cite{z1,z2}, the following recommendations were formulated:
\begin{itemize}
    \item For the algorithm to operate as expected as a replacement for a median filter, it is necessary to ensure the uniqueness of elements in the array being processed.
    \item In register $D$, which contains a superposition of array elements, one must ensure the absence of data-unfilled qubits on the top and bottom. Unfilled qubits are those that contain zeros for all possible array values.
    \item Increasing the size of the sliding window in image processing leads to better filtering quality.
\end{itemize}
From the results presented in Section \ref{exp}, it can be seen that the algorithm performs slightly worse than the median filter, while providing a gain in computational complexity from $O(n*M*log_2M)$ to $O(n)$.

\section*{Acknowledgements}
The work of K.R.Zakharova was supported by the Ministry of Science and Higher Education of the Russian Federation (agreement 075-15-2025-344 dated 29/04/2025 for Saint Petersburg Leonhard Euler International Mathematical Institute at PDMI RAS).

\end{document}